\documentclass[11pt]{article}
\usepackage{graphicx}
\usepackage{amsmath}
\usepackage{multirow}
\usepackage{amssymb}
\usepackage{amsthm}
\usepackage{natbib}
\usepackage{moreverb}
\usepackage[T1]{fontenc}
\usepackage{gensymb}
\usepackage{pdfpages}
\usepackage[figuresright]{rotating}
\usepackage{geometry}
\usepackage{bbm}
\usepackage{setspace}
\usepackage{csquotes}
\renewcommand{\P}{{\text{pr}}}
\newcommand{\E}{{{E}}}

\newcommand{\bH}{{{H}}}

\newcommand{\m}{{{m}}}

\newcommand{\bx}{{{x}}}
\newcommand{\bZ}{{{Z}}}
\newcommand{\bz}{{{z}}}

\newcommand{\bd}{{{d}}}

\newcommand{\var}{{\text{var}}}
\newcommand{\bY}{{{Y}}}
\newcommand{\bV}{{{V}}}

\newcommand{\Z}{{{Z}}}
\renewcommand{\bH}{{{H}}}
\newcommand{\bD}{{{D}}}
\newcommand{\bI}{{{I}}}

\newcommand{\M}{{{M}}}

\newcommand{\cF}{{\mathcal{F}}}

\newtheorem{theorem}{Theorem}

\newtheorem{corollary}{Corollary}
\newtheorem{lemma}{Lemma}
\newtheorem{fact}{Fact}

\theoremstyle{definition}
\newtheorem{condition}{Condition}

\theoremstyle{remark}

\title{Regression assisted inference for the average treatment effect in paired experiments}
\author{Colin B. Fogarty \thanks{Operations Research and Statistics Group, MIT Sloan School of Management, Massachusetts Institute of Technology, Cambridge MA 02142 (e-mail: \texttt{cfogarty@mit.edu})}}
\date{}
\begin{document}
\maketitle
\begin{abstract}
In paired randomized experiments individuals in a given matched pair may differ on prognostically important covariates despite the best efforts of practitioners. We examine the use of regression adjustment as a way to correct for persistent covariate imbalances after randomization, and present two regression assisted estimators for the sample average treatment effect in paired experiments. Using the potential outcomes framework, we prove that these estimators are consistent for the sample average treatment effect under mild regularity conditions even if the regression model is improperly specified. Further, we describe how asymptotically conservative confidence intervals can be constructed. We demonstrate that the variances of the regression assisted estimators are at least as small as that of the standard difference-in-means estimator asymptotically. Through a simulation study, we illustrate the appropriateness of the proposed methods in small and moderate samples. The analysis does not require a superpopulation model, a constant treatment effect, or the truth of the regression model, and hence provides a mode of inference for the sample average treatment effect with the potential to yield improvements in the power of the resulting analysis over the classical analysis without imposing potentially unrealistic assumptions.
\end{abstract}

\section{Introduction}\label{sec:intro}
The use of paired experiments has historically been limited to cases in which pairs are formed using a small number of binary or categorical covariates. In practice, there are often continuous covariates which are also believed to be predictive for the potential outcomes under treatment and control. \citet{gre04} introduced a form of multivariate matching before randomization which assigns a group of $2n$ experimental units to $n$ pairs to minimize the within-pairs covariate distance, hence improving covariate balance for many variables at the same time. Once the pairs are established, exactly one subject in each pair is randomly assigned to the treatment. Fisherian inference for no treatment effect and Neymanian inference for the average treatment effect can then proceed with respect to the randomization distribution generated by this paired design; see, for example, \citet{ros02cov} and \citet{ima08}.

While mitigable through the method of \citet{gre04}, persistent within-pair discrepancies on the basis of continuous covariates of interest are unavoidable. These discrepancies may, in turn, yield chance imbalances on the basis of these covariates in any given randomization. A common strategy for accounting for remaining imbalances is covariance adjustment. \citet{fre08} and \citet{lin13} investigate the impact of regression adjustment on inference for the sample average treatment effect in completely randomized experiments. Neither correctness of the fitted regression model nor a superpopulation model (that units are drawn independently and identically distributed from a larger population) are assumed in their analysis. Instead, inference proceeds using the physical act of randomization as its sole justification, and asymptotic calculations consider sequences of experiments of increasing size without specifying how the experimental units were sampled. In this context, \citet{lin13} discusses how regression adjustment can yield a consistent estimator of the sample average treatment effect whose asymptotic variance is no larger than the standard difference-in-means estimator regardless of the truth of the underlying model. The performed inference is then agnostic to the truth of the underlying model and instead leverages the adjustment as an algorithmic fit in an attempt to yield a more efficient estimator. See also \citet{aro13} and \citet{blo16} for related work with completely randomized experiments.

In the context of paired experiments, \citet{ros02cov} describes how covariance adjustment can be utilized to yield exact tests under the assumption of a constant treatment effect.  \citet[\S 10]{imb15} discuss how under a superpopulation model, linear regression can yield a consistent estimator for the average treatment effect in the superpopulation. In what follows, we show that covariance adjustment can also be leveraged for inference on the sample average treatment effect in a paired experiment in a manner that is agnostic to the truth of the fitted model.


\section{A paired randomized experiment with covariates}\label{sec:notation}
The $i$th of $n$ independent pairs contains one individual assigned the treatment, denoted as $Z_{ij} = 1$, and one who receives the control, $Z_{ij'} = 0$, such that $Z_{i1} + Z_{i2} = 1$. Individual $j$ in matched set $i$ has a $P$-dimensional vector of measured covariates $\bx_{ij} = (x_{ij1},...,x_{ijP})$. Each individual has a potential outcome under treatment, $r_{Tij}$, and under control, $r_{Cij}$, $i=1,...,n; j = 1,2$. The pair of potential outcomes $(r_{Tij}, r_{Cij})$, and with it the treatment effect $\tau_{ij} = r_{Tij}-r_{Cij}$, is not jointly observable for any individual. Instead, we observe the response $R_{ij} = r_{Tij}Z_{ij} +r_{Cij}(1-Z_{ij})$ for each individual \citep{ney23, rub74}. Quantities dependent on the assignment vector such as $\Z = (Z_{11}, Z_{12},..., Z_{n2})$ and ${R} = (R_{11}, R_{12},...,R_{n2})$ are random, whereas $\cF = \{(r_{Tij}, r_{Cij}, x_{ij}), i = 1,...,n, j = 1,2\}$ contains fixed quantities. In a paired experiment, $\P(\bZ = \bz\mid\cF) = \P(\bZ = \bz)=  2^{-n}$, and $\P(Z_{ij} = 1 \mid \cF) = \P(Z_{ij} = 1)= 1/2$ for $i=1,...,n$, $j=1,2$.

Write $\bd_{i} = f(\bx_{i1}) - f(\bx_{i2})$ for some function $f: \mathbb{R}^P \rightarrow \mathbb{R}^{K_D}$ with $K_D$ fixed as the difference in transformed covariates between unit $1$ and unit $2$ in matched pair $i$. Similarly write $\m_i = \{g(\bx_{i1})+g(x_{i2})\}/2 - n^{-1}\sum_{i'=1}^n\{g(\bx_{i'1})+g(x_{i'2})\}/2$ for some function $g:\mathbb{R}^{P} \rightarrow \mathbb{R}^{K_M}$, as the average of the transformed covariates in matched pair $i$ relative to the mean across matched pairs. For example, if $x$ is scalar then choosing $f(x) = (x,x^2)$, $g(x) = \exp(x)$ yields $K_D = 2$, $K_M = 1$. Guidance for the choice of the functions $f(\cdot)$ and $g(\cdot)$, which may differ, is given in \S \ref{sec:reg}. Write $\bD$ for the $n\times K_D$ matrix whose $i$th row equals $\bd_{i}^T$, and write $M$ for the  $n\times K_M$ matrix whose $i$th row equals $m_{i}^T$. In what follows we require that $K_D + K_M +1 < n$.  Let $\bH_D = \bD(\bD^T\bD)^{-1}\bD$ be the hat matrix for $\bD$, i.e. the orthogonal projection of $\mathbb{R}^{K}$ onto the column space of $\bD$, and let $\bH_M$ be that hat matrix for $\M$. Define $V_i = 2Z_{i1} - 1$ such that $\E(V_i\mid\cF)=0$, and let $\bV$ be the $n\times n$ matrix with $V_i$ on the $i$th diagonal and zeroes on the off-diagonal. The matrix consisting of the treated-minus-control differences of the $f$-transformed covariates can be written as $\bV\bD$, with hat matrix $\bV\bD(\bD^T\bV\bV\bD)^{-1}\bD\bV = \bV\bH_D\bV$ since $V_i^2 = 1$. Finally, write $A$ for the $n\times (K_D + K_M)$ matrix whose $i$th row is $a_i^T = (V_id_{i1},...,V_id_{iK_D}, m_{i1}, ...,m_{iK_M})^T$, and write $\bH_A$ for the corresponding hat matrix.

\section{The sample average treatment effect}\label{sec:ATE}

Let $\ell_{ij} = (r_{Tij} + r_{Cij})/2$ be the average, or level, of the potential outcomes for individual $j$ in pair $i$, and let $\Delta_i = (\tau_{i1} + \tau_{i2})/2$ be the average of the treatment effects in pair $i$. The observed treated-minus-control difference in responses in pair $i$ is	$Y_{i} = \Delta_i + V_i(\ell_{i1} - \ell_{i2})$.  Write $\ell_j = (\ell_{1j},...,\ell_{nj})$ for $j = 1,2$, $\Delta = (\Delta_1,...,\Delta_n)$, and $\bY = (Y_1,...,Y_n)$. The sample average treatment effect in a paired experiment with $n$ pairs is 
\begin{align*}\bar{\tau}_n &= \frac{1}{2n}\sum_{i=1}^n(\tau_{i1} + \tau_{i2}) = \frac{1}{n}\sum_{i=1}^n{\Delta}_i.\end{align*}

Henceforth we will write $\bar{\tau} = \bar{\tau}_n$, but the dependence of this and other causal estimands on $n$ should be kept in mind. The classical unbiased estimator for $\bar{\tau}$ in a paired experiment is the observed average of the paired differences $\hat{\tau}_C = \bar{Y} = n^{-1}\sum_{i=1}^n Y_{i}$.  The variance for $\hat{\tau}_C$ across randomizations is $\var(\hat{\tau}_C\mid\cF) = n^{-2}\sum_{i=1}^n(\ell_{i1} -\ell_{i2})^2$,
which is unknown in practice because it depends on the missing potential outcomes. \citet{ima08} shows that the classical estimator of the variance of the difference in means in a paired study, $S_C^2 = \{n(n-1)\}^{-1}\sum_{i=1}^n(Y_i - \hat{\tau}_C)^2$,
is always an upper bound on $\var(\hat{\tau}_C\mid\cF)$ in expectation. Furthermore, under mild regularity conditions $nS^2_C - \var({n}^{1/2}\hat{\tau}_C\mid\cF)$ converges to a nonnegative value in probability, and asymptotically conservative inference can be achieved by employing a normal approximation with $S_C^2$ in place of $\var(\hat{\tau}_C\mid\cF)$. 
\section{Regression assisted estimators}\label{sec:reg}
We consider the intercept coefficients from a regression of $\bY$ on $\bV\bD$, $\hat{\tau}_{R1}$, and of $\bY$ on $\bV\bD$ and $M$, $\hat{\tau}_{R2}$, as estimators for $\bar{\tau}$. These can be expressed as
\begin{align*}
\hat{{\tau}}_{R1} = \frac{{e}^T(\bI-\bV\bH_D\bV)\bY}{{e}^T(\bI-\bV\bH_D\bV)e};\;\;\;
\hat{{\tau}}_{R2} = \frac{{e}^T(\bI-\bH_A)\bY}{{e}^T(\bI-\bH_A)e},\end{align*} where $\bI$ is the identity matrix, and ${e}$ is the vector containing $n$ ones (i.e., the intercept column). As discussed in \citet[][\S 10]{imb15}, a regression on $\bV\bD$ encodes a belief that the difference between $f$-transformed covariates in matched pair $i$, $d_i$, may be predictive of the difference in the levels of the potential outcomes $\ell_{i1}-\ell_{i2}$. A regression including both $\bV\bD$ and $M$ encodes a belief that the difference in $f$-transformed covariates in matched pair $i$ is predictive of the difference in the levels, and that the relative level of the $g$-transformed covariates $m_i$ is predictive of the average treatment effect in pair $i$, ${\Delta}_i$. As such, the function $f(\cdot)$ giving rise to $d_i$, should, ideally, be chosen to best reflect the relationship between the transformed differences $d_i = f(x_{i1}) - f(x_{i2})$ and the difference in level of the potential outcomes $\ell_{i1}-\ell_{i2}$. Similarly, the function $g(\cdot)$ should be chosen to best reflect the relationship between the pairwise average of the transformed covariates, $(g(x_{i1}) + g(x_{i2}))/2$, and the pairwise average of the treatment effects $\Delta_i$.

Under a constant effect model $\hat{\tau}_{R1}$ is an unbiased estimator for $\bar{\tau}$, and inference for $\tau$ can be conducted using the techniques described in \citet{ros02cov}.  Without the assumption of additivity, neither $\hat{\tau}_{R1}$ nor $\hat{\tau}_{R2}$ are unbiased for $\bar{\tau}$. Nonetheless, we now demonstrate that both $\hat{\tau}_{R1}$ and $\hat{\tau}_{R2}$ can be used to facilitate inference on $\bar{\tau}$ without constant effects. Moving forward, we impose two regularity conditions.

\begin{condition}\textit{(Bounded Fourth Moments)} \label{cond:1}. There exists a $C <\infty$ such that, for all $n$, $n^{-1}\sum_{i=1}^n\Delta_i^4 < C$, $n^{-1}\sum_{i=1}^n(\ell_{i1}-\ell_{i2})^4 < C$, $n^{-1}\sum_{i=1}^nd_{ik}^4 < C$, $k=1,...,K_D$, and $n^{-1}\sum_{i=1}^nm_{ik'}^4 < C$, $k'=1,...,K_M$.
\end{condition}
\begin{condition}\textit{(Ces\`{a}ro Summability)}. \label{cond:2}  $n^{-1}\sum_{i=1}^n\Delta_i$, $n^{-1}\sum_{i=1}^n\Delta_i^2$, and $n^{-1}\sum_{i=1}^n(\ell_{i1}-\ell_{i2})^2$ converge to finite limits as $n\rightarrow \infty$. Further, $n^{-1}\sum_{i=1}^n(\ell_{i1} - \ell_{i2})d_{ik}$ and $n^{-1}\sum_{i=1}^n\Delta_im_{ik}$ converge to finite limits for $k = 1,...,K_D$, $k' = 1,...,K_M$ as $n\rightarrow \infty$. Finally, $n^{-1}D^TD$ and $n^{-1}M^TM$ converges to a finite, invertible matrices. Call these limits $\Sigma_D$ and $\Sigma_M$ respectively.
\end{condition} Let $\beta_D$ be the limit of $(D^TD)^{-1}D^T(\ell_1-\ell_2)$, and let $\beta_M$ be the limit of $(M^TM)^{-1}M^T\Delta$.
\begin{theorem}\label{thm:1}  Define $\hat{\tau}_{R*} = n^{-1}\sum_{i=1}^n(Y_i - V_i\bd_i^T\beta_D)$, and suppose Conditions \ref{cond:1} and \ref{cond:2} hold. Then, as $n\rightarrow \infty$ with $K_D$ and $K_M$ fixed and conditional upon $\cF$, both $n^{1/2}(\hat{\tau}_{R1}  - \hat{\tau}_{R*})$ and $n^{1/2}(\hat{\tau}_{R2}  -\hat{\tau}_{R*})$ converge in probability to zero. Furthermore, $n^{1/2}(\hat{\tau}_{R*} - \bar{\tau})$ converges in distribution to a Gaussian random variable with mean zero and variance
\begin{align*} \sigma^2_{R*} &= \underset{n\rightarrow\infty}{\lim}n^{-1}\sum_{i=1}^n(\ell_{i1}-\ell_{i2})^2 - \beta_D^T\Sigma_D\beta_D.\end{align*}
\end{theorem}
Theorem \ref{thm:1} characterizes several useful properties of the regression assisted estimators $\hat{\tau}_{R1}$ and $\hat{\tau}_{R2}$ under Conditions \ref{cond:1} and \ref{cond:2}. First, $\hat{\tau}_{R1}$ and $\hat{\tau}_{R2}$ are consistent estimators of $\bar{\tau}$. Second, $n^{1/2}(\hat{\tau}_{R1} - \bar{\tau})$ and $n^{1/2}(\hat{\tau}_{R2} - \bar{\tau})$ are asymptotically equivalent, with a common Gaussian limiting distribution. Finally, note that the first term in $\sigma^2_{R*}$ is precisely the asymptotic variance of $n^{1/2}\hat{\tau}_C$ given $\cF$. Hence, the asymptotic variances of the regression assisted estimators are no larger than that of the classical difference-in-means estimator, as $\beta_D^T\Sigma_D\beta_D\geq 0$ by positive semi-definiteness of $\Sigma_D$. Better choices of $f(\cdot)$ impact the magnitude of $\beta_D^T\Sigma_D\beta_D$ and hence the degree of variance reduction.

Perhaps surprisingly, $\hat{\tau}_{R1}$ and $\hat{\tau}_{R2}$ have the same asymptotic variance. To see why, the average value of the covariates in a given matched pair does not vary across randomizations; however, the difference between the covariates for the treated and control individuals does, taking values $d_{i}$ and $-d_{i}$ with equal probability. If the slope coefficients on the level of covariates $m_i$ have limits as $n\rightarrow\infty$ (as guaranteed by Conditions \ref{cond:1} and \ref{cond:2}), then asymptotically the contribution of $m_i$ to the prediction of $Y_i$ also does not vary across randomizations, and hence does not contribute to the variance of the estimator. The choice of the function $g(\cdot)$ thus plays no role in improving efficiency. In light of this, is there motivation for including the level of the covariates when estimating $\bar{\tau}$? As we now demonstrate, inclusion of $M$ in the regression allows for the construction of variance estimators that are less conservative then those derived from a regression excluding $M$.

\section{Enabling Neyman-style inference}\label{sec:neyman}
As described in \S \ref{sec:ATE} inference using the classical difference-in-means estimator, $\hat{\tau}_C$, typically proceeds using an upper bound on the variance. If the estimator is a consistent upper bound, inference performed is then asymptotically conservative if a Gaussian reference distribution is asymptotically justified. As will be demonstrated, the nominal variance estimators for the intercept coefficients derived from linear model theory with fixed design and homoskedastic errors can be employed towards this end. Under homoskedasticity the classical variance estimators for the intercept coefficients $\hat{\tau}_{R1}$ and $\hat{\tau}_{R2}$ take the form
\begin{align*}
S^2_{R1}= \frac{(Y-\hat{\tau}_{R1}e)^T(I-VH_DV)(Y - \hat{\tau}_{R1}e)}{(n-K_D-1)e^T(I-VH_DV)e};\;\;\;
S^2_{R2} = \frac{(Y-\hat{\tau}_{R2}e)^T(I-H_A)(Y - \hat{\tau}_{R2}e)}{(n-K_D - K_M-1)e^T(I-H_A)e}.
\end{align*} 
See \S D of the appendix for a derivation. The developments that follow also apply if these estimators are replaced with heteroskedasticity consistent standard errors \citep{lon00}. 

\begin{theorem}\label{thm:2}
Under Conditions \ref{cond:1} and \ref{cond:2} and conditional upon $\cF$, $nS^2_{R1} - \var(n^{1/2}\hat{\tau}_{R*}\mid\cF)$ converges in probability to 
\begin{align*}\underset{n\rightarrow\infty}{\lim}n^{-1} \sum_{i=1}^n(\Delta_i - \bar{\tau})^2 \geq 0. \end{align*} 

Under the same conditions, $nS^2_{R2} - \var(n^{1/2}\hat{\tau}_{R*}\mid\cF)$ converges in probability to 
 \begin{align*} \underset{n\rightarrow\infty}{\lim}n^{-1} \sum_{i=1}^n(\Delta_i - \bar{\tau})^2 - \beta_M^T\Sigma_M\beta_M = \underset{n\rightarrow\infty}{\lim}n^{-1}(\Delta - \bar{\tau}e)^T(I-H_M)(\Delta - \bar{\tau}e) \geq 0.\end{align*} 
\end{theorem}
\begin{corollary}\label{cor:1}
Under these conditions, $nS^2_{R1} - nS^2_{R2}$ converges in probability to \begin{align*}\beta_M^T\Sigma_M\beta_M  \geq 0.\end{align*}
\end{corollary}
Theorem \ref{thm:2}, along with Corollary \ref{cor:1}, indicate that both $nS^2_{R1}$ and $nS^2_{R2}$ will be consistent upper bounds for  $\var(n^{1/2}\hat{\tau}_{R*}\mid\cF)$, but that $nS^2_{R2}$ will be asymptotically no larger than $nS^2_{R1}$.  Asymptotic equality is attained under an additive treatment effect model, where both $nS^2_{R1}$ and $nS^2_{R2}$ are consistent for $\var(n^{1/2}\hat{\tau}_{R*}\mid\cF)$. Unlike $nS^2_{R1}$, $nS^2_{R2}$ can also be consistent for $\var(n^{1/2}\hat{\tau}_{R*}\mid\cF)$ if the relative level of the covariates in a pair $m_i$ is perfectly predictive of the pairwise average treatment effects $\Delta_i$, hence highlighting the role of the function $g(\cdot)$ used in defining $M$.  In combination with Theorem \ref{thm:1}, Theorem \ref{thm:2} indicates that on asymptotic grounds $\hat{\tau}_{R2}$ should be preferred over $\hat{\tau}_{R1}$. Confidence intervals for $\bar{\tau}$ of the form $\hat{\tau}_{R1} \pm \Phi^{-1}(1-\alpha/2)S_{R1}$ and $\hat{\tau}_{R2} \pm \Phi^{-1}(1-\alpha/2)S_{R2}$ will be asymptotically conservative as desired. Confidence intervals for the sample average treatment effect $\bar{\tau}$ constructed using $S_{R2}$ will be no longer than those constructed using $S_{R1}$, and hypothesis tests for the null $H_0: \bar{\tau} = \bar{\tau}_0$ through the test statistic $(\hat{\tau}_{R2} - \bar{\tau}_0)/S_{R2}$ and a Gaussian reference distribution will be more powerful than that using $(\hat{\tau}_{R1} - \bar{\tau}_0)/S_{R1}$ while asymptotically maintaining the desired size.

\section{A simulation study}
\subsection{Linear regression under a nonlinear truth}
Our study contains two simulation settings inspired by the functions used in \citet{kan07}. The $s$th of $S$ samples of $n$ pairs is generated by first sampling covariates $w_{ijp}$, $p=1,..,4$ that are unknown to the researcher for each of the $2n$ study participants. For each $p$, we stipulate that $w_{i1p}$ is Gaussian distributed with mean zero and variance 1, $w_{i2p}$ Gaussian distributed with mean $w_{i1p}$ and variance 1/4, and $w_{ijp}$ is independent of $w_{ijp'}$ for $p'\neq p$. Potential outcomes under treatment and control are simulated through $r_{Tij} = \mu_T(w_{ij}) + \varepsilon_{ij}$ and $r_{Cij} = \mu_C(w_{ij}) + \varepsilon_{ij}$, with $\varepsilon_{ij}$ standard normal. The sample average treatment effect for sample $s$ is $\bar{\tau}^{(s)} =  (2n)^{-1}\sum_{i=1}^n\sum_{j=1}^2\left\{\mu_T(w_{ij}) - \mu_C(w_{ij})\right\}$.  The settings vary in the functions $\mu_T(\cdot)$ and $\mu_C(\cdot)$. The two possibilities for $\mu_T(\cdot)$ and $\mu_C(\cdot)$ are (1) parallel response surfaces with $\mu_T(w_{ij}) = \mu_C(w_{ij}) =  27.4w_{ij1} + 13.7(w_{ij2}+w_{ij3}+w_{ij4})$, implying constant treatment effects; and (2) nonparallel response surfaces with $\mu_T(w_{ij}) = 27.4w_{ij1} + 13.7(w_{ij2}+w_{ij3}+w_{ij4})$ and $\mu_C(w_{ij}) = 13.7(w_{ij1} + w_{ij2})+ 3w_{ij3}+27.4w_{ij4}$, implying heterogeneous treatment effects.

The revealed covariates $x_{ij}$ for each individual are complicated functions of $w_{ij}$,
\begin{align*} x_{ij1} = \exp(w_{ij1}/2);\;\;\;\; x_{ij2} =  w_{ij2}/\{1 + \exp(w_{ij1})\} + 10\;\\ x_{ij3} = \{(w_{ij1}w_{ij3})/25+ 0.6\}^3;\;\;\;\; x_{ij4}=(w_{ij2}+w_{ij4} + 20)^2.\end{align*} For each randomization, the researcher computes $\hat{\tau}_{R1}$ and $\hat{\tau}_{R2}$ from a multiple regression using the observed covariates $x_{ij}$ instead of $w_{ij}$ and setting $f(x_{ij}) = g(x_{ij}) = x_{ij}$, in so doing fitting a linear model using covariates that have a highly nonlinear relationship with the potential outcomes. With the experimental units and observed covariates established, we conduct $B$ randomizations wherein we assign the treatment to exactly one unit in each pair. For each randomization, we calculate $\hat{\tau}_C$, $\hat{\tau}_{R1}$ and $\hat{\tau}_{R2}$ based on the observed responses. We then construct normal-based 95\% confidence intervals for the experiment-specific sample average treatment effect $\bar{\tau}^{(s)}$ using $S^2_C$, $S^2_{R1}$, and $S^2_{R2}$. 

\subsection{Results for two sets of experimental units}
Figure \ref{fig:1} illustrates the results of $B=1000$ randomizations for two samples with $n=500$ individuals, one from each of the specifications for the response functions (parallel and nonparallel). The histograms show the across-randomization distributions of the three estimators minus the true value of $\bar{\tau}^{(s)}$. As predicted by Theorem \ref{thm:1} all three estimators are well approximated by normal distributions. Further, we see that the true dispersions of $\hat{\tau}_{R1}$ and $\hat{\tau}_{R2}$ are closely aligned, but that both estimators have smaller variance than $\hat{\tau}_{RC}$ despite the fact that the regression model was misspecified. The second plots for each setting compare the length of normal-based 95\% confidence intervals constructed using the true variance for the estimators to the typical length of 95\% intervals constructed using the sample estimators for those variances, $S^2_{R1}$, $S^2_{R2}$, and $S^2_{C}$ respectively.  Here, the conclusions of Theorem \ref{thm:2} come to bear. Under constant effects (left), the expected sample-based interval overlaps with the intervals based on the true variance, as with constant effects the three variance estimators are consistent. Consequently, the corresponding confidence intervals had coverage of roughly 95\%. With heterogeneous treatment effects (right) we instead see that typical intervals based on the sample at hand are wider than those based on the true variance, as under effect heterogeneity the estimators are instead consistent upper bounds. This conservativeness led to confidence intervals that, for all estimators, had coverage of 100\% in the $B=1000$ randomizations. Comparing $\hat{\tau}_{R1}$ to $\hat{\tau}_{R2}$, we see that while the ideal confidence intervals based on the true variance are quite similar the corresponding sample-based intervals differ, with those based on $\hat{\tau}_{R2}$ being roughly 2/3 the length of those based on $\hat{\tau}_{R1}$ while still exceeding their coverage guarantee. Despite the regression model being misspecified, effect modification was leveraged by $S^2_{R2}$ to yield narrower intervals. In \S E.1 of the appendix, we present detailed results for this simulation study with $n=25, 50, 100$, and $500$. Therein, we see that the predictions of Theorem \ref{thm:1} and \ref{thm:2} provide appropriate guidance even in moderately sized samples. 

\begin{figure}
\begin{center}
\includegraphics[scale = .4]{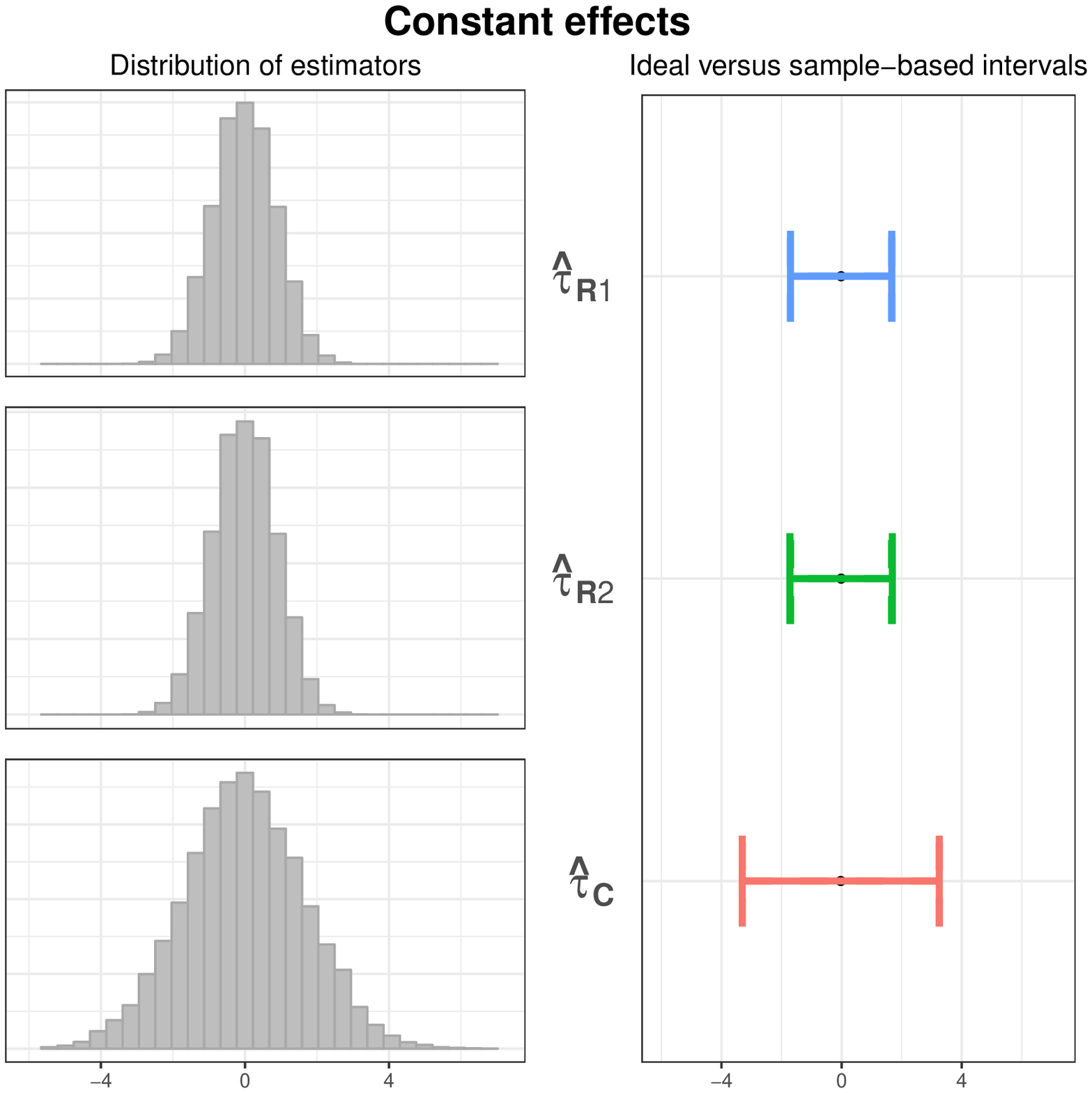}
\hspace{.1 in}
\includegraphics[scale = .4]{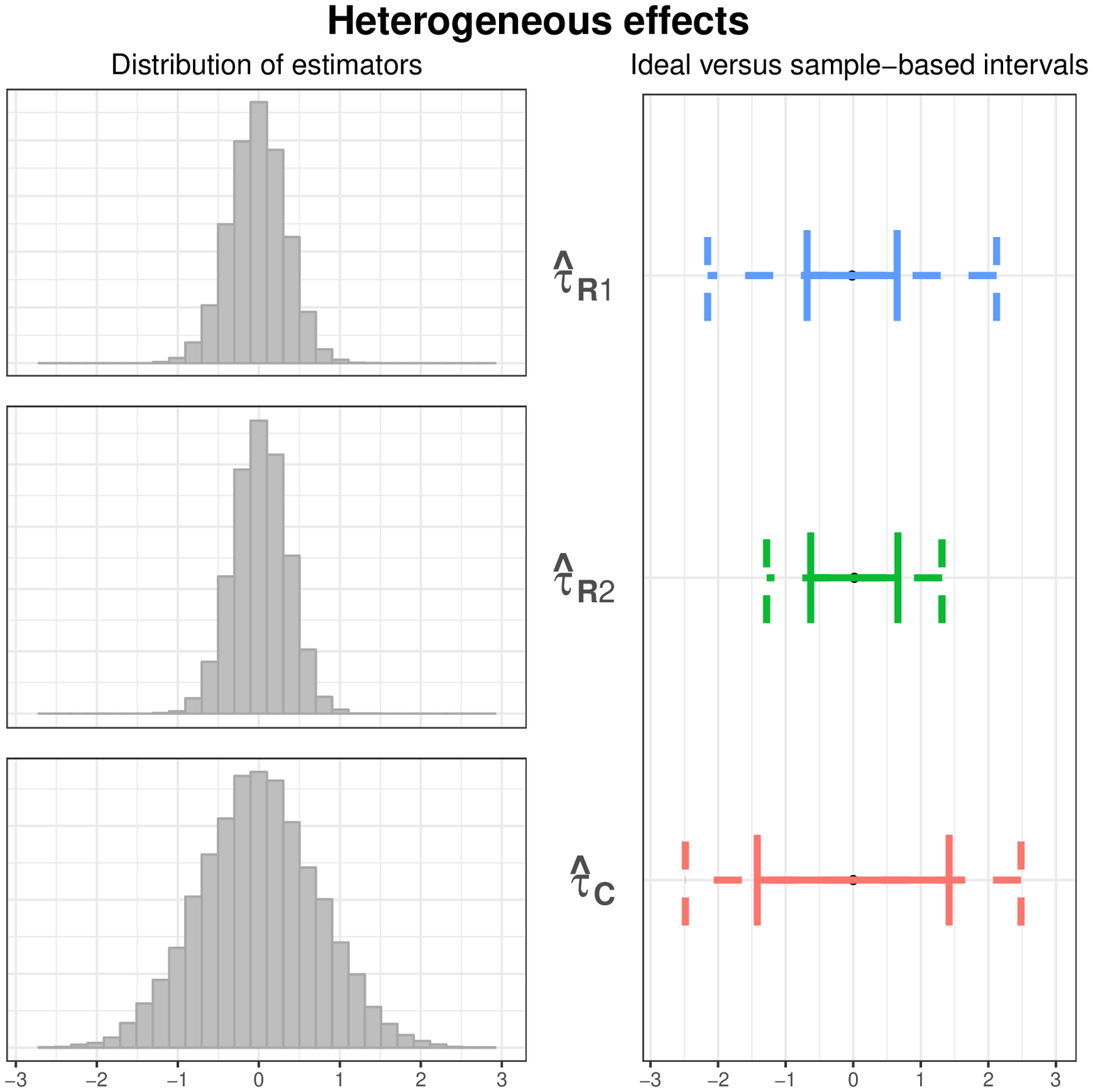}

\caption{A graph showing, for each estimator, its simulated distribution, normal-based 95\% range based on true variances (solid), and the approximation to the 95\% range based on estimated variances (dashed) with constant (left) and heterogeneous (right) treatment effects.}
\label{fig:1}
\vspace{-.1 in}
\end{center}
\end{figure}

\section{Superpopulation inference after regression adjustment}
The improvement in inference by means of $\hat{\tau}_{R2}$ and $S^2_{R2}$ is not without qualification. Suppose that one instead considered a superpopulation model wherein $n$ pairs of individuals, and hence their potential outcomes and covariates, were drawn at random from an infinite population and that inference is desired for the expectation of the treatment effect in that superpopulation, call it $\bar{\tau}^{(P)}$. Inference is no longer conducted conditional upon $\cF$ as described previously in this article, but instead must account for variation across realizations of $\cF$. In this setting, $\hat{\tau}_{R1}$ and $\hat{\tau}_{R2}$ remain consistent for $\bar{\tau}^{(P)}$, are asymptotically equivalent, and have true variance that is no larger than that of the difference-in-means estimator $\hat{\tau}_C$. The limiting variances $\var(n^{1/2}\hat{\tau}_{R1})$ and $\var(n^{1/2}\hat{\tau}_{R2})$ are increased by $\var(n^{1/2}\bar{\tau})$, as the sample average treatment effect is itself a random variable under this formulation with $\var(\bar{\tau})$ equal to the expectation of $\{n(n-1)\}^{-1}\sum_{i=1}^n(\Delta_i-\bar{\tau})^2$. Arguments parallel to those in \S 5.3 of a 2017 unpublished paper by the author (arXiv:1706.06469) show that $nS^2_{R1}$ is a consistent estimator of $\var(n^{1/2}\hat{\tau}_{R1})$, while $nS^2_{R2}$ (which achieves its lower probability limit by exploiting effect modification) is an underestimate of $\var(n^{1/2}\hat{\tau}_{R2})$, in turn yielding anti-conservative inference. For superpopulation inference using $\hat{\tau}_{R2}$, one should instead employ the corrected variance estimator $S^2_{R2,P} = S^2_{R2} + n^{-1}\hat{\beta}_M^T\hat{\Sigma}_M\hat{\beta}_M$, where $\hat{\Sigma}_M = (n-1)^{-1}M^TM$ is the sample covariance matrix for the centered covariates $M$ and $\hat{\beta}_M$ is the vector of regression coefficients corresponding to the columns of $M$ from a regression of $Y$ on $A = (VD, M)$ along with an intercept. Defined in this manner, $nS^2_{R2,P}$ is consistent for $\var(n^{1/2}\hat{\tau}_{R2})$. The importance of this correction is highlighted in a simulation study presented in \S E.2 of the appendix. In light of Corollary 1, we see that asymptotically $nS^2_{R2,P}$ is increased by precisely the discrepancy between $nS^2_{R2}$ and $nS^2_{R1}$. In fact, $nS^2_{R1}-nS^2_{R2,P}$ converges in probability to zero, hence eliminating the inferential advantage held by $\hat{\tau}_{R2}$ over $\hat{\tau}_{R1}$ asymptotically. Whether finite-sample properties of inference based on $\hat{\tau}_{R2}$ versus $\hat{\tau}_{R1}$ would lead one estimator to be preferred in a superpopulation setting remains an area for further research. 

\section{Discussion}
This work has focused on the use of agnostic linear regression (i) to yield an estimator of the sample average treatment effect with improved asymptotic efficiency over the difference-in-means, and (ii) to furnish conservative standard error estimators for regression-adjusted estimators which dominate the conventional standard error for the difference-in-means. A natural and important extension of this work would be to consider other forms of regression adjustment in paired experiments by leveraging semiparametric theory, aligning with the approach taken for completely randomized designs by \citet{zha08} and \citet{tsi08}. The focus on linear regression  serves, in part, to provide justification for what remains the most common form of covariance adjustment utilized by practitioners under minimal assumptions about the manner in which the data were generated.

In completely randomized experiments, \citet{aro14} presents sharp variance estimators for the difference-in-means estimator of the sample average treatment effect without covariance adjustment. These results may furnish improvements over the variance estimators considered by \citet{lin13} for regression adjustment in completely randomized designs; however, the natural extension to block-randomized designs requires at least two treated and two control individuals in each block, a feature absent in paired experiments. Nonetheless, the variance estimator $S^2_{R2}$ provides a means of improving the estimator $S^2_{R1}$ when the target of estimation is the sample average treatment effect.


\appendix
\section{Facts and Lemmas}
The following fact about the matrix $M$, true by construction, will be useful moving forward.
\begin{fact} \label{fact:1}$e^Tm_{k'}=0$ and $(I-H_M)e = e$ due to centering of the covariates $m_{k'}, k' = 1,...,K_M$.
\end{fact}
We now prove three lemmas which will be utilized in the proofs of our main results.
\begin{lemma}\label{lemma:1} $n^{-1}A^TA$ converges in probability to a block-diagonal matrix with $\underset{n \rightarrow\infty}{\lim}n^{-1}D^TD$ in the upper-left block of size $K_D\times K_D$, $\underset{n\rightarrow\infty}{\lim}n^{-1}M^TM$ in the lower-right block of size $K_M\times K_M$, and zeroes in the remaining entries.\end{lemma}
\begin{proof}
It suffices to show that $n^{-1}\sum_{i=1}^nV_id_{ik}m_{ik'}$ converges in probability to zero for any $k=1,...,K_D$, $k' = 1,...,K_M$. Trivially, $\E(n^{-1}\sum_{i=1}^nV_id_{ik}m_{ik'}\mid\cF) = 0$. Through Chebyshev's inequality, we can show that $\var(n^{-1}\sum_{i=1}^nV_id_{ik}m_{ik'}\mid\cF) \rightarrow 0$ to complete the proof.  The proof utilizes Condition \ref{cond:1}, along with the Cauchy-Schwarz inequality.
\begin{align*}\var\left(n^{-1}\sum_{i=1}^nV_id_{ik}m_{ik'}\mid\cF\right)&=n^{-2}\sum_{i=1}^n(d_{ik})^2(m_{ik'}^2)\\ 
&\leq n^{-2}
\left(\sum_{i=1}^nd_{ik}^4\right)^{1/2}\left(\sum_{i=1}^nm_{ik'}^4\right)^{1/2}  \leq C/n
\end{align*}
by Condition \ref{cond:1}, which tends to zero as $n\rightarrow \infty$.
\end{proof}
\begin{lemma} \label{lemma:2}$n^{-1}\sum_{i=1}^nV_iY_id_{ik}$ converges in probability to $\underset{n\rightarrow\infty}{\lim}n^{-1}\sum_{i=1}^n(\ell_{i1}-\ell_{i2})d_{ik}$ for any $k=1,...,K_D$. Similarly, $n^{-1}\sum_{i=1}^nY_im_{ik'}$ converges in probability to $\underset{n\rightarrow\infty}{\lim} n^{-1}\sum_{i=1}^n\Delta_im_{ik'}$ for any $k'=1,...,K_M$, and $n^{-1}\sum_{i=1}^IV_id_{ik}$ converges in probability to zero for any $k=1,...,K$.
\end{lemma}
\begin{proof}
We prove the result for $n^{-1}\sum_{i=1}^nV_iY_id_{ik}$; the remaining proofs are analogous. First, $E(n^{-1}\sum_{i=1}^nV_iY_id_{ik}\mid\cF) = n^{-1}(\ell_{i1}-\ell_{i2})d_{ik}$. Similar derivations to those in Lemma \ref{lemma:1} yield by Condition \ref{cond:1} that $\var(n^{-1}\sum_{i=1}^nV_iY_id_{ik}\mid\cF)$ converges to zero, which along with Chebyshev's inequality and existence of the limiting value by Condition \ref{cond:2} completes the proof.
\end{proof}
\begin{lemma}\label{lemma:3} $n^{-1}e^T(I-VH_DV)e$ and $n^{-1}e^T(I-H_A)e$ converge in probability to zero.
\end{lemma}
\begin{proof}
We prove the result for $n^{-1}e^T(I-H_A)e$, and in the process we also prove the result for $n^{-1}e^T(I-VH_DV)e$. 
\begin{align*}
n^{-1}e^T(I-H_A)e &= 1 - n^{-1}\sum_{i=1}^na_i(A^TA)^{-1}A^Te\\
&= 1 - n^{-1}\sum_{i=1}^nV_id_i\left\{(n^{-1}D^TD)^{-1}n^{-1}D^TVe\right\} + o_p(1),
\end{align*}
where the $o_p(1)$ term comes from Fact \ref{fact:1}, Condition \ref{cond:2}, and Lemma \ref{lemma:1}. By Condition \ref{cond:2} and Chebyshev's inequality, $n^{-1}D^TVe = n^{-1}\sum_{i=1}^nV_id_i$ converges to a $K_D$ dimensional vector with zeroes in every entry, while $n^{-1}D^TD$ converges to a finite limit, thus completing the proof.
\end{proof}

\section{Proof of Theorem 1}

Recall the definition of $\hat{\tau}_{R*}$ as \begin{align*}
\hat{\tau}_{R*} &= n^{-1}\sum_{i=1}^n(Y_i - V_i\bd_i^T\beta_D).
\end{align*}We first prove the result for $n^{1/2}\hat{\tau}_{R2}$. Let $\hat{\beta}_D$ and $\hat{\beta}_M$ be vectors of length $K_D$ and $K_M$ corresponding to the sample slopes for the covariates $VD$ and $M$ respectively from a regression of $Y$ on $(e,A)$. By Lemmas \ref{lemma:1} and \ref{lemma:2}, $\hat{\beta}_D$ and $\hat{\beta}_M$ converge in probability to $\beta_D$ and $\beta_M$. Along with Lemmas \ref{lemma:2} and \ref{lemma:3}, the discrepancy $n^{1/2}(\hat{\tau}_{R2} - \hat{\tau}_{R*})$ is then of the form 
\begin{align*} n^{1/2}(\hat{\tau}_{R2} - \hat{\tau}_{R*})  &= n^{-1/2}\sum_{i=1}^nm_{i}^T\hat{\beta}_M + n^{-1/2}\sum_{i=1}^n V_id_i^T(\hat{\beta}_D - \beta_D)+ o_p(1) = o_p(1)\end{align*} since by construction $\sum_{i=1}^nm_{i}^T\hat{\beta}_{M} = 0$ by Fact 1, $n^{1/2}\sum_{i=1}^nV_id_i$ is $O_p(1)$ under Conditions \ref{cond:1} and \ref{cond:2}, and $(\hat{\beta}_D - \beta_D)$ is $o_p(1)$. Hence, $n^{1/2}(\hat{\tau}_{R*} - \hat{\tau}_{R*})$ converges in probability to zero as desired.  By Conditions \ref{cond:1} and \ref{cond:2}, Lyapnuov's central limit theorem holds for $n^{1/2}\hat{\tau}_{R*}$ at $\delta=2$, whose asymptotic variance follows from a straightforward calculation. To prove the result for $\hat{\tau}_{R1}$, simply disregard mention of $\beta_M$ and $\hat{\beta}_M$.

\section{Proof of Theorem 2}
Recall the definitions for $S^2_{R1}$ and $S^2_{R2}$:
\begin{align*}
S^2_{R1} &= (n-K_D-1)^{-1}\frac{(Y-\hat{\tau}_{R1}e)^T(I-VH_DV)(Y - \hat{\tau}_{R1}e)}{e^T(I-VH_DV)e},\\
S^2_{R2} &= (n-K_D - K_M-1)^{-1}\frac{(Y-\hat{\tau}_{R2}e)^T(I-H_A)(Y - \hat{\tau}_{R2}e)}{e^T(I-H_A)e}.
\end{align*}


We prove the remark for $S^2_{R2}$. Note first that by Lemma \ref{lemma:3}, $n\{e^T(I-H_A)e\}^{-1}$ converges in probability to 1. Re-expressing $S^2_{R2}$ using standard identities for residuals from linear regression
\begin{align}
nS^2_{R2}  &= n^{-1}\{1+o_p(1)\}(Y-\hat{\tau}_{R2}e)^T(I-H_M)(Y-\hat{\tau}_{R2}e)\\ & - n^{-1}\{1+o_p(1)\}(Y-\hat{\tau}_{R2}e)^T(I-H_M)VD\{D^TV(I-H_M)VD\}^{-1}DV(I-H_M)(Y - \hat{\tau}_{R2}e).\end{align}
We begin by assessing term (1). By Condition \ref{cond:2}, Lemma \ref{lemma:2} and Theorem 1, (1) converges in probability to
\begin{align*}
\underset{n\rightarrow\infty}{\lim}n^{-1}\sum_{i=1}^n(\Delta_i - \bar{\tau})^2 + n^{-1}\sum_{i=1}^n(\ell_{i1}-\ell_{i2})^2 -  n^{-1}(\Delta- \bar{\tau}e)^TH_M(\Delta - \bar{\tau}e). \end{align*}

We now assess term (2). Recall that  any term of the form $n^{-1}M^TVD$ converges in probability to a matrix of all zeroes as described in the proof of Lemma \ref{lemma:1}, and recall that $H_M = M(M^TM)^{-1}M^T$. Hence, by Condition \ref{cond:2}, Lemma \ref{lemma:2}, and Theorem 1, (2) converges in probability to
\begin{align*}
 -\underset{n\rightarrow\infty}{\lim}n^{-1}(\ell_1-\ell_2)^TH_D(\ell_1-\ell_2) = -\beta_D^T\Sigma_D\beta_D.
\end{align*}
Meanwhile, $\var(n^{1/2}\hat{\tau}_{R2}\mid\cF) {\rightarrow} \underset{n\rightarrow\infty}{\lim}n^{-1}\sum_{i=1}^n(\ell_{i1}-\ell_{i2})^2 - \beta_D^T\Sigma_D\beta_D$ as given in Theorem 1. Subtracting these two quantities, $nS^2_{R2} - \var(n^{1/2}\hat{\tau}_{R2}\mid\cF)$ converges in probability to
\begin{align*}
 \underset{n\rightarrow\infty}{\lim}n^{-1}\sum_{i=1}^n(\Delta_i - \bar{\tau})^2 -  n^{-1}(\Delta- \bar{\tau}e)^TH_M(\Delta - \bar{\tau}e)=  \underset{n\rightarrow\infty}{\lim}n^{-1}\sum_{i=1}^n(\Delta_i - \bar{\tau})^2 -  \beta_M^T\Sigma_M\beta_M.
\end{align*} 

The proof for $S^2_{R1}$ follows by replacing $I-H_M$ with $I$ and $\hat{\tau}_{R2}$ with $\hat{\tau}_{R1}$ in (1) and (2) above.

\section{Deriving the Expressions for the Variance Estimators}
We now derive the form of the variance estimators presented in Section 5. We do so for $S^2_{R2}$; the derivation for $S^2_{R1}$ is analogous. In the classical homoskedastic fixed-$X$ regression setting, a regression of $Y$ on a fixed matrix $X = (e, A)$ yields a true variance for the intercept coefficient
\begin{align*}
\var(\hat{\beta}_0) &= \sigma^2\frac{e^T(I-H_A)e}{(e^T(I-H_A)e)^2}  =\sigma^2\frac{1}{e^T(I-H_A)e}.
\end{align*}
One then replaces $\sigma^2$ with the mean squared error, $\hat{\sigma}^2$, to yield the variance estimator. Let $H_X$ be the projection matrix associated with $X = (e,A)$. $I-H_X$ can be attained by iterative projections, that is
\begin{align*}
I - H_X &= I - \frac{(I-H_A)ee^T(I-H_A)}{e^T(I-H_A)e}.
\end{align*}
Recalling that $\hat{\tau}_{R2} = Y^T(I-H_A)e/(e^T(I-H_A)e)$, the sum of square errors is then of the form
\begin{align*}
SSE &= Y^T\left( I - \frac{(I-H_A)ee^T(I-H_A)}{e^T(I-H_A)e}\right)Y\\
&= Y^TY - \hat{\tau}_{R2}e^T(I-H_A)e\hat{\tau}_{R2}.
\end{align*}
Further, $\hat{\tau}_{R2}e^T(I-H_A)e\hat{\tau}_{R2} = e^T(I-H_A)ee^T(I-H_A)Y/(e^T(I-H_A)e) = e^T(I-H_A)Y$. Hence
\begin{align*}
SSE &= Y^TY - \hat{\tau}_{R2}e^T(I-H_A)e\hat{\tau}_{R2}\\
&= (Y - \hat{\tau}_{R2})(I - H_A)(Y-\hat{\tau}_{R2}).
\end{align*}The form for $S^2_{R2} = SSE/((n-K_D-K_M-1)e^T(I-H_A)e)$ then follows.
\section{Additional simulation results}
\subsection{Confidence intervals and standard errors for the sample average treatment effect}
In each of the settings described in \S 6.1 of the article, we simulate $S=1000$ sets of experimental units each containing $n$ matched pairs. Within each sample, we conduct $B=1000$ randomizations. We conducted our simulation study with experiments of size $n=25, 50, 100$ and $n=500$. Given the experiment sizes considered, the extent to which the results presented in this work are applicable for experiments of these sizes is far but certain at the onset of this investigation.

The first three columns characterize the actual coverage properties of 95\% confidence intervals for the sample average treatment effect constructed using the standard formulae for the standard errors of $\hat{\tau}_{C}$, $\hat{\tau}_{R1}$, and $\hat{\tau}_{R2}$ respectively. The second three columns compare the relative values of these standard errors based on these three estimators (and hence, the relative width of confidence intervals across experiments). The final two column shows the relative values for the true root mean squared errors of the regression estimators $\hat{\tau}_{R1}$ and $\hat{\tau}_{R2}$ relative to the classical estimator $\hat{\tau}_C$. The initial number provided in each entry is the median value across the $1000$ sets of experimental units, and the endpoints of the interval provided in parenthesis are the 2.5\% and 97.5\% quantiles across the $1000$ samples.

\begin{table}
\begin{center}
\caption{Classical and regression assisted average treatment effect estimators with a misspecified response function, sample average treatment effect}{
\begin{tabular}{l c c c c c c c c}
&\multicolumn{3}{c}{True Coverage of 95\% CIs} &  \multicolumn{3}{c}{100 $\times$ Stand. Error Ratios}& \multicolumn{2}{c}{100 $\times$ RMSE Ratios}\\ [1 pt]
	&$C$ & $R1$ &$R2$	&$R1:C$&$R2:C$	&  $R2:R1$ &$R1:C$&$R2:C$\\
Parallel&	&&&&&&&\\
\multirow{2}{*}{$n=25$} & 95 &95& 94 & 55	&55& 100	&  54 & 55 \\
 &(94, 96) & (94, 96) & (89, 97) & (37, 77) & (37, 77) & (91, 107) & (37, 78) & (38, 85)\\
\multirow{2}{*}{$n=50$}& 95 &95& 94 & 54	&53& 99	&  54 & 55 \\
 &(94, 96) & (94, 96) & (92, 96) & (41, 69) & (41, 67) & (94, 101) & (41, 69) & (41, 72)\\
 \multirow{2}{*}{$n=100$}& 95 &95& 95 & 54	&54& 99	&  54 & 55 \\
 &(94, 96) & (94, 96) & (93, 96) & (45, 65) & (44, 64) & (96, 100) & (44, 65) & (45, 67)\\
 \multirow{2}{*}{$n=500$}& 95 &95& 95 & 54	&54& 100	&  54 & 55 \\
 &(94, 96) & (94, 96) & (94, 96) & (50, 60) & (50, 60) & (99, 100) & (50, 61) & (50, 61)\\
\multicolumn{2}{l}{Nonparallel}	&&&&&&&\\
 \multirow{2}{*}{$n=25$}& 100 &100& 99 & 90	&54& 60	&  81 & 63 \\
 &(99, 100) & (99, 100) & (96, 100) & (77, 90) & (40, 71) & (43, 83) & (55, 125) & (43, 98)\\
\multirow{2}{*}{$n=50$}& 100 &100& 100 & 86	&53& 61	&  64 & 53 \\
 &(99, 100) & (99, 100) & (99, 100) & (78, 94) & (43, 63) & (50, 74) & (48, 88) & (39, 74)\\
\multirow{2}{*}{$n=100$}& 100 &100& 100 & 86	&53& 62	&  55 & 48 \\
&(100, 100) & (100, 100) & (100, 100) & (80, 90) & (46, 60) & (53, 71) & (45, 69) & (39, 61)\\
 \multirow{2}{*}{$n=500$}& 100 &100& 100 & 85	&53& 62	&  45 & 44 \\
&(100, 100) & (100, 100) & (100, 100) & (82, 97) & (50, 56) & (53, 71) & (40, 51) & (39, 50)\\
\end{tabular}}\label{tab:results}
\end{center}
\end{table}

In the parallel response setting with $\mu_T(\cdot) = \mu_C(\cdot)$, the performed inference will be asymptotically exact rather conservative since here the treatment effects are constant. As Table 1 indicates, for all values of $n$ the median coverage of intervals constructed using $\hat{\tau}_C$ and $\hat{\tau}_{R1}$ for $\bar{\tau}^{(s)}$ over the $S=1000$ sets of experimental units was 95\%, with coverage falling in the range (94\%, 96\%) for 95\% of the simulated experiments. For $\hat{\tau}_{R2}$, the median coverage was 94\% at $n=25$, with a 95\% range of (89\%, 97\%) indicating a more drastic deviation from claimed coverage across randomizations in some sets of experimental units with $n$ small. The reason for this is not small sample bias in the estimator $\hat{\tau}_{R2}$ itself, but rather a systematic downward bias of $S_{R2}$ relative to the true value of $\var(\hat{\tau}_{R2}\mid\cF)^{1/2}$. In fact, the 95\% range for $E(S_{R2}|\cF)/\var(\hat{\tau}_{R2}\mid\cF)^{1/2}$ across samples was (0.81, 1.05) with $n=25$, which may justify the use of heteroskedasticy-consistent standard errors such as HC3 in smaller samples \citep{lon00}. The coverage improves quickly as $n$ increases, and presents little reason for alarm once $n=100$. The median ratio of both the standard errors and the root mean squared errors from $\hat{\tau}_{R1}$ and $\hat{\tau}_{R2}$ relative to $\hat{\tau}_C$ are slightly over 1/2 across experiments, indicating the potential benefits from estimation and inference using regression adjustment in this setting. As predicted by Theorem 2, the ratios of the standard error estimates based on $\hat{\tau}_{R1}$ and $\hat{\tau}_{R2}$ are clustered about 1 due to the treatment effect being additive, indicating that even in finite samples there is little to be gained from using $\hat{\tau}_{R2}$ over $\hat{\tau}_{R1}$ under additivity.

For the settings with a nonparallel response surface, intervals constructed will be conservative since the treatment effect is no longer additive. Indeed, we see that for all values of $n$ the median coverage across experiments was either 99\% or 100\% for confidence intervals based on all three estimators. In these settings, we also observe the potential benefits of using the estimator $\hat{\tau}_{R2}$. For $\hat{\tau}_{R1}$, the standard errors used to construct confidence intervals were only slightly smaller than those constructed using $\hat{\tau}_{C}$, with the median ratio of the two across experiments decreasing from 0.9 for $n=25$ to 0.85 at $n=500$. On the other hand, the median ratio of interval widths constructed using $\hat{\tau}_{R2}$ relative to those using $\hat{\tau}_{C}$ was slightly over 1/2,  while the median ratio relative to those constructed based on $\hat{\tau}_{R1}$ was $0.6$. To emphasize, this result does not stem from an inherently lower variance for the estimator $\hat{\tau}_{R2}$, as  reflected in the ratios of the true root mean square errors for estimation using $\hat{\tau}_{C}$ relative to $\hat{\tau}_{R1}$ and relative to $\hat{\tau}_{R2}$ shown in Table \ref{tab:results}. Indeed, the mean squared errors for $\hat{\tau}_{R1}$ and $\hat{\tau}_{R2}$ are equal in the limit by Theorem 1. Rather, as Corollary 1 describes, this stems from $S^2_{R2}$ being a tighter upper bound for $\var(\hat{\tau}_{R2}\mid\cF)$ than $S^2_{R1}$ is for $\var(\hat{\tau}_{R1}\mid\cF)$. Hence, intervals based on $\hat{\tau}_{R2}$ are shorter than those constructed using $\hat{\tau}_{R1}$ while maintaining the proper coverage.

\subsection{Coverage in a superpopulation simulation}

In this section, we modify the simulation setting described in \S 6.1 of the manuscript to instead assess coverage for the population average treatment effect. We keep the same functional forms for the response functions under treatment and control, and maintain the same misspecified form for the covariates available to the practitioner. In the $s$th of $S=10000$ sets of experimental units, we simply conduct one random allocation to treatment or control in accordance with the paired design. We then compute the estimators in question, corresponding standard errors, and confidence intervals. Rather than attempting to cover $\bar{\tau}^{(s)}$, the intervals are now assessed based on covering $E\{\bar{\tau}^{(s)}\} = 0$. For $\hat{\tau}_{R2}$, we compute intervals based on both $S^2_{R2}$ and its superpopulation modification $S^2_{R2,P} = S^2_{R2} + n^{-1}\hat{\beta}_M^T\hat{\Sigma}_M\hat{\beta}_M.$ We do so for $n=25$ and $n=1000$.

Table \ref{tab:resultssuper} present the results. In the parallel response setting, the results are quite similar to those presented in \S E.1, where we instead constructed confidence intervals for the sample average treatment effect. The reason is that the treatment effect was constant for all individuals, such that $\var\{\bar{\tau}^{(c)}\} = 0$ in the superpopulation simulation. In the case of constant effects $S^2_{R1}$, $S^2_{R2}$, and $S^2_{R2,P}$ are asymptotically equivalent, and $S^2_{R2}$ can be used to yield valid confidence intervals in this superpopulation setting. The estimators $\hat{\tau}_{R1}$ and $\hat{\tau}_{R2}$ remain asymptotically equivalent, and continue to improve efficiency over $\hat{\tau}_C$.

\begin{table}
\begin{center}
\caption{Classical and regression assisted average treatment effect estimators with a misspecified response function, population average treatment effect. The table shows coverage of 95\% normal-based intervals, the ratio of the average estimated standard error to the true standard deviation of the estimator (in parentheses), and the ratio of the standard deviations of the estimators.}{
\begin{tabular}{l c c c c c c c}
&\multicolumn{4}{c}{True Coverage of 95\% CIs} & \multicolumn{3}{c}{100 $\times$ RMSE Ratios}\\ [1 pt]
&\multicolumn{4}{c}{(100 $\times$ Estimated SE / True SE)} & \multicolumn{3}{c}{}\\

	&$C$ & $R1$ &$R2$&$R2,P$	&  $R2:R1$ &$R1:C$&$R2:C$\\
Parallel&	&&&&&\\
\multirow{2}{*}{$n=25$}& 95.0 &95.2& 93.8 & 96.8	&\multirow{2}{*}{107}& \multirow{2}{*}{54.1}	&  \multirow{2}{*}{58.1}\\
&  (98.0) & (99.7) & (92.7) & (104) &  & & \\
\\
 \multirow{2}{*}{$n=1000$}& 94.8 &95.0& 94.9 & 95.0	&\multirow{2}{*}{100}&\multirow{2}{*}{54.1}	&  \multirow{2}{*}{54.3}  \\
 &(99.3) & (99.5) & (99.2) & (99.5) &  &  & \\
 \\
\multicolumn{2}{l}{Nonparallel}	&&&&&&\\
 \multirow{2}{*}{$n=25$}& 95.3 &94.3& 77.2 & 95.6	&\multirow{2}{*}{97.9}& \multirow{2}{*}{92.2}	&  \multirow{2}{*}{90.2}  \\
 &  (99.0) & (96.3) & (59.3) & (99.8) &  & & \\
 \\
\multirow{2}{*}{$n=1000$}& 95.1 &95.2& 78.2 & 95.1	&\multirow{2}{*}{99.9}& \multirow{2}{*}{84.7}&  \multirow{2}{*}{84.6}\\
&(100) & (100) & (62.8) & (101) & & &\\
\end{tabular}}\label{tab:resultssuper}
\end{center}
\end{table}

As the nonparallel case illustrates, validity of $S^2_{R2}$ in the superpopulation context described is unique to the case of constant treatment effects. In this simulation, effect heterogeneity was present. $S^2_{R2}$ attains its smaller value relative to $S^2_{R1}$ by means of exploiting effect heterogeneity. Yet as we see, $S^2_{R1}$ yields confidence intervals that meet their coverage guarantee, while $S^2_{R2}$ creates confidence intervals that fall well short. As the table presents, the ratio of the average standard error $S_{R2}$ to the true standard deviation of $\hat{\tau}_{R2}$ is far below 1, further indicative of $S^2_{R2}$ being anticonservative. The downwards bias does not disappear asymptotically, as the simulation study with $n=1000$ seeks to emphasize. The corrected version, $S^2_{R2,P}$, adds back a consistent estimate of the downwards bias in $S^2_{R2}$, which reestablishes proper coverage for intervals using $\hat{\tau}_{R2}$. Once again, $\hat{\tau}_{R1}$ and $\hat{\tau}_{R2}$ behave quite similarly in terms of estimator variance (indicative of their asymptotic equivalence), and both dominate the classical difference-in-means $\hat{\tau}_C$. 
\newpage
\subsection{A numerical example}
The following \texttt{R} code simulates an experiment in the nonparallel simulation setting from \S 6.1 of the manuscript with $I=25$, showing the numerical values for the treatment effect estimators and the corresponding standard errors.

\begin{verbatimtab}
library(MASS)
#Nonparallel Setting, n=25 Pairs
set.seed(01562)
#Set number of pairs
n = 25
#Simulate (unknown) covariates W, error, and potential outcomes
w1 = mvrnorm(n, c(0,0,0,0), diag(4))
w2 = w1 + mvrnorm(n, c(0,0,0,0), .25*diag(4))
e1 = rnorm(n, 0, 1)
e2 = rnorm(n, 0, 1)
rt1 = w1%*%(c(27.4,13.7,13.7, 13.7)) + e1 
rc1 = w1%*%(c(13.7,13.7,3, 27.4)) + e1
rt2= w2%*%(c(27.4,13.7,13.7, 13.7)) + e2
rc2= w2%*%(c(13.7,13.7,3, 27.4)) + e2

#Form Delta (pairwise average of treatment effects) 
tau1 = rt1-rc1
tau2 = rt2-rc2
Delta = (tau1+tau2)/2

#Form eta (difference in levels)
l1 = (rt1+rc1)/2
l2 = (rt2+rc2)/2
eta = l1-l2

#Form the observed covariates as nonlinear functions of W
X1 = cbind(exp(w1[,1]/2), w1[,2]/(1+exp(w1[,1])) + 10, 
           (w1[,1]*w1[,3]/25 + .6)^3, (w1[,2] + w1[,4] + 20)^2)
X2 = cbind(exp(w2[,1]/2), w2[,2]/(1+exp(w2[,1])) + 10, 
           (w2[,1]*w2[,3]/25 + .6)^3, (w2[,2] + w2[,4] + 20)^2)

#Form D and M based on X
D = X1-X2
D = as.matrix(D)
M = (X1+X2)/2 - (rep(1, n))%*%t(colMeans(rbind(X1, X2)))
KD = ncol(D)
KM = ncol(M)

#Sample and Population ATE in simulation
SATE = (sum(rt1-rc1) + sum(rt2-rc2))/(2*I)
PATE = 0

#Randomly assign treatment across pairs; form V
Z = runif(n) < .5
V = 2*Z-1

#Here are the treated-minus-control response differences
Y = Delta + V*eta

#Here are the treated-minus-control covariate differences
VD = diag(V)%*%D

#Form the classical estimator and its standard error
tauhat_C = mean(Y)
SE_C = sd(Y)/sqrt(n)

#Now, form tauhat_R1 with its standard error
reg1 = lm(Y~VD)
tauhat_R1 = reg1$coef[1]
SE_R1 = summary(reg1)$coef[1,2]


#Now, tauhat_R2. The difference is the inclusion of M in
#the regression
reg2 = lm(Y~VD+M)
tauhat_R2 = reg2$coef[1]
SE_R2 = summary(reg2)$coef[1,2]

#Form the superpopulation version of SE_R2
betaM = reg2$coef[6:9]
SE_R2P = sqrt(SE_R2^2 + t(betaM)%*%cov(M)%*%betaM/n)

#Compare the point estimates and standard errors
compare25 = rbind(c(tauhat_C, tauhat_R1, tauhat_R2, tauhat_R2), 
                  c(SE_C, SE_R1, SE_R2, SE_R2P))
colnames(compare25) = c("tauhat_C", "tauhat_R1", "tauhat_R2", 
                        "tauhat_R2P")
rownames(compare25) = c("Estimate", "Standard Error")
compare25

#Output
#                tauhat_C tauhat_R1 tauhat_R2 tauhat_R2P
# Estimate       3.640938 -1.884071 -2.728688  -2.728688
# Standard Error 5.483459  3.935346  2.966589   4.077766
\end{verbatimtab}

\bibliographystyle{apalike}
\bibliography{bibliography.bib}

\begin{thebibliography}{}

\bibitem[Aronow et~al., 2014]{aro14}
Aronow, P.~M., Green, D.~P., and Lee, D.~K. (2014).
\newblock {Sharp bounds on the variance in randomized experiments}.
\newblock {\em The Annals of Statistics}, 42(3):850--871.

\bibitem[Aronow and Middleton, 2013]{aro13}
Aronow, P.~M. and Middleton, J.~A. (2013).
\newblock A class of unbiased estimators of the average treatment effect in
  randomized experiments.
\newblock {\em Journal of Causal Inference}, 1(1):135--154.

\bibitem[Bloniarz et~al., 2016]{blo16}
Bloniarz, A., Liu, H., Zhang, C.-H., Sekhon, J., and Yu, B. (2016).
\newblock Lasso adjustments of treatment effect estimates in randomized
  experiments.
\newblock {\em Proceedings of the National Academy of Sciences},
  113(27):7383--7390.

\bibitem[Freedman, 2008]{fre08}
Freedman, D.~A. (2008).
\newblock On regression adjustments to experimental data.
\newblock {\em Advances in Applied Mathematics}, 40(2):180--193.

\bibitem[Greevy et~al., 2004]{gre04}
Greevy, R., Lu, B., Silber, J.~H., and Rosenbaum, P. (2004).
\newblock Optimal multivariate matching before randomization.
\newblock {\em Biostatistics}, 5(2):263--275.

\bibitem[Imai, 2008]{ima08}
Imai, K. (2008).
\newblock Variance identification and efficiency analysis in randomized
  experiments under the matched-pair design.
\newblock {\em Statistics in {M}edicine}, 27(24):4857--4873.

\bibitem[Imbens and Rubin, 2015]{imb15}
Imbens, G.~W. and Rubin, D.~B. (2015).
\newblock {\em Causal inference in statistics, social, and biomedical
  sciences}.
\newblock Cambridge University Press.

\bibitem[Kang and Schafer, 2007]{kan07}
Kang, J.~D. and Schafer, J.~L. (2007).
\newblock Demystifying double robustness: A comparison of alternative
  strategies for estimating a population mean from incomplete data.
\newblock {\em Statistical {S}cience}, pages 523--539.

\bibitem[Lin, 2013]{lin13}
Lin, W. (2013).
\newblock Agnostic notes on regression adjustments to experimental data:
  Reexamining {Freedman's} critique.
\newblock {\em The Annals of Applied Statistics}, 7(1):295--318.

\bibitem[Long and Ervin, 2000]{lon00}
Long, J.~S. and Ervin, L.~H. (2000).
\newblock Using heteroscedasticity consistent standard errors in the linear
  regression model.
\newblock {\em The American Statistician}, 54(3):217--224.

\bibitem[Neyman, 1923]{ney23}
Neyman, J. (1923).
\newblock On the application of probability theory to agricultural experiments.
  {E}ssay on principles. {S}ection 9 (in {P}olish).
\newblock {\em Roczniki Nauk Roiniczych}, X:1--51.
\newblock Reprinted in Statistical Science, 1990, 5(4):463-480.

\bibitem[Rosenbaum, 2002]{ros02cov}
Rosenbaum, P.~R. (2002).
\newblock Covariance adjustment in randomized experiments and observational
  studies.
\newblock {\em Statistical Science}, 17(3):286--327.

\bibitem[Rubin, 1974]{rub74}
Rubin, D.~B. (1974).
\newblock Estimating causal effects of treatments in randomized and
  nonrandomized studies.
\newblock {\em Journal of Educational Psychology}, 66(5):688--701.

\bibitem[Tsiatis et~al., 2008]{tsi08}
Tsiatis, A.~A., Davidian, M., Zhang, M., and Lu, X. (2008).
\newblock Covariate adjustment for two-sample treatment comparisons in
  randomized clinical trials: A principled yet flexible approach.
\newblock {\em Statistics in medicine}, 27(23):4658--4677.

\bibitem[Zhang et~al., 2008]{zha08}
Zhang, M., Tsiatis, A.~A., and Davidian, M. (2008).
\newblock Improving efficiency of inferences in randomized clinical trials
  using auxiliary covariates.
\newblock {\em Biometrics}, 64(3):707--715.

\end{thebibliography}

\end{document}